\documentclass[aoas,preprint]{imsart}

\RequirePackage[OT1]{fontenc}
\RequirePackage{amsthm,amsmath, amstext, amsfonts, verbatim}
\RequirePackage[numbers]{natbib}
\RequirePackage[colorlinks,citecolor=blue,urlcolor=blue]{hyperref}

\date{}
\usepackage{graphicx}

\usepackage{enumerate}

\startlocaldefs
\numberwithin{equation}{section}
\theoremstyle{plain}
\newtheorem{thm}{Theorem}[section]
\newtheorem{defn}{Definition}[section]
\endlocaldefs

\begin{document}

\begin{frontmatter}

\title{Detection of Phase Shift Events}
\runtitle{Detection of Phase Shift Events}

\begin{aug}
\author{\fnms{William} \snm{Marshall}\thanksref{m1}\ead[label=e1]{wjmarsha@uwaterloo.ca}}
\and
\author{\fnms{Paul} \snm{Marriott}\thanksref{m1}\ead[label=e2]{pmarriott@uwaterloo.ca}}

\runauthor{W. Marshall et al.}

\affiliation{Department of Statistics and Actuarial Science, University of Waterloo\thanksmark{m1}}

\address{University of Waterloo\\
200 University Ave W \\
Waterloo, On, Canada \\
\printead{e1}\\
\phantom{E-mail:\ }\printead*{e2}}

\end{aug}

\begin{abstract}

We consider the problem of change-point estimation of  the   instantaneous phase of an observed time series. Such change points, or \textit{phase shifts}, can be markers of information transfer in complex systems;   their analysis  occurring in geology, biology and physics, but most notably in neuroscience.  We develop two non-parametric approaches to this problem: the cumulative summation (CUSUM) and phase derivative (PD)  estimators. In general the CUSUM estimator has higher power for identifying single shift events, while the PD estimator has better temporal resolution for multiple ones. A system of weakly coupled R\"{o}ssler attractors provides an  application in which there are high levels of systematic and time-dependent noise. Shift identification is also performed on beta-band activity from electroencephalogram recordings of a visual attention task, an unsupervised application which requires high temporal resolution.

\end{abstract}



\end{frontmatter}

\section{Introduction}

This paper  develops statistical methodology which is  suitable for tackling the  problem of (instantaneous)  phase shift identification from observable time series, such as in electroencephalogram (EEG) recordings. The problem is formulated in terms of  finding a  change-point in the  instantaneous phase variable. Since these phase variables are computed  from complex systems which often  include significant amounts of noise, we focus  on developing methods which are robust to such noise, but also  have the temporal resolution necessary for many application areas. 

The analysis of  phase change points   has applications in many  fields. In neuroscience, phase shift behaviour is present on small-scale (cellular) analysis of spike bursting data~\citep{Izhikevich_Ch10}. It has also been applied to other complex biological oscillators, such as the Circadian rhythm~\citep{Mead1992} and in cardiac/respiratory systems~\citep{Balocchi2004}. In geology, directional (circular) data is used to identify faults, pinchouts and other geological features~\citep{Taner1979}; angular data of the Earths rotation was used to characterize the Chandler Wobble~\citep{Gibert2008}, while  in physics, phase shift identification has applications to interferometry~\citep{EspanaBoquera1996, Grigorenko1999} and in characterizing nano-scale surfaces~\citep{Martinez2006}.  In this paper, though, we focus, for concreteness and because of their intrinsic  importance, on applications in EEG analysis  and dynamical systems.

Phase synchronization is a characteristic feature of neural assemblies, which   are  some of the  fundamental units of information processing in the brain~\citep{Braitenberg1985}. Periods of synchronization (or phase locking) occurs when  neurons associate with one another, forming  a neural assembly  to accomplish a task. A spontaneous desynchronization (or phase shift) represents dissolution of the neural assembly, and the re-organization of neural resources for the next task. When neural activity is viewed on a large scale, such as in EEG recordings, desynchronization may manifest as a phase shift in the instantaneous phase of the signal.

Due to its extremely good time resolution  EEG  is  used in brain-computer interface (BCI) applications, where the goal is to provide an augmentative communication system which can interpret spontaneous brain activity~\citep{Bai2008, Wolpaw2000}. In such applications  real-time analysis of the EEG signal is vital and this need for online computation will be highlighted  in the analysis below.

Measures of phase synchronization are also used to quantify interactions within complex (chaotic) dynamic systems of oscillators~\citep{Boccaletti2002}. One such system is the coupled R\"{o}ssler attractor, where synchronization has been observed that  is independent of the corresponding amplitude~\citep{Osipov1997, Rosenblum1996, Pikovsky1985}. The experiments of \citet{Rosenblum1996} show that strong coupling results in complete synchronization; however, during weak coupling the system displayed dynamic transitions between phase locking and phase shifting behaviour.  The existence of weak or intermediate coupling strengths is also a feature of systems which display self organized criticality (SOC), a class of dynamics which result in apparent  power-law distributions~\citep{Bak1987}. 

The goal of this paper is to develop robust methods for phase shift identification from observables where we  formulate the problem as a change-point problem. In Section \ref{methods} we describe instantaneous phase estimation, as well as the  change-point  methods which will be used. Section \ref{Osci} contains simulations of simple oscillators, both with and without phase shift behaviour, which are used to assess the proposed methods in a controlled environment. In Section \ref{Ross}, methods are applied to data generated from a system of coupled R\"{o}ssler attractors which have chaotic dynamics and spontaneous phase shift behaviour. In Section \ref{EEG}, the analysis is applied to EEG recordings where we  identify phase shifts in the beta band of participants during a visual task.

\section{Methods}
\label{methods}

\subsection{Instantaneous Phase}
\label{instant}

Phase can be considered as   the value which defines the initial state of an oscillator, and many  linear techniques, such as the Fourier transform, assume a constant phase value over the each window of estimation~\citep{Bingham1967}. Nonlinear methods are used to estimate time-dependent measures of  {\it instantaneous phase}  which capture moment-to-moment changes. Instantaneous phase can be defined using either the Hilbert~\citep{Gabor1946} or Wavelet transformations, both have been shown to produce similar results~\citep{LeVanQuyen2001}, here we use the Hilbert definition. In practice, the complex demodulation algorithm is used to estimate the instantaneous phase of an observed signal~\citep{Bingham1967,Goodman1960}.
\begin{defn}
For a given $\delta$Hz low-pass filter $\mathbb{H}[\cdot]$, the instantaneous phase of an observed signal $x_t$ in the frequency band $(\omega-\delta, \omega+\delta)$ is, 
\begin{eqnarray}
\label{CDemod}
\hat{\phi}_t = \tan^{-1}\left( \frac{\mathbb{H}[x_t\sin\{-\omega t\}]}{\mathbb{H}[x_t\cos\{-\omega t\}]} \right).
\end{eqnarray}
\end{defn}

The $\delta$Hz low-pass filter isolates the frequency band of interest, this ensures that the calculated instantaneous phase can be meaningfully interpreted as the phase of the $(\omega-\delta, \omega+\delta)$ component of the signal~\citep{Chavez2006}. Also note that removing  any  discontinuities due to the inverse tangent function is called {\em phase straightening} and will be done for applications in this paper.

\subsubsection{Theoretical  Properties}
\label{theory}

Consider  the  instantaneous phase estimator in the case of an oscillator with additive i.i.d.~noise and using a simple first-order filter. The exponentially weighted moving average (EWMA) is a tractable digital low-pass filter which is used to investigate properties of the estimator,
\begin{eqnarray}
\label{EWMA}
\mathbb{H}[x_t] = (1-\alpha)\sum_{i=0}^t \alpha^{i}x_{t-i}, \quad \alpha \in (0,1).
\end{eqnarray}
If the parameter is set to $\alpha=e^{2\pi f_c/T}$, then $f_c$ is the -3 dB cutoff point for the filter (i.e. frequency components above $f_c$ has been reduced by a factor of at least $10^{-3}$).

\begin{thm}
\label{T1}
Suppose $x_t$ is a noisy oscillator sampled at T Hz with frequency $f_0$ Hz ($\omega = 2\pi f_0/T$) and constant phase $\phi$ i.e.,  $x_t = \sin{(\omega t + \phi)} + \epsilon_t$,
with $E(\epsilon_t)=0$. Let $y_t=\mathbb{H}[x_t\sin\{-\omega t\}]$ and $\tilde{y_t}=\mathbb{H}[x_t\cos\{-\omega t\}]$ be the two components in the complex demodulation estimate of $\hat{\phi}_t$ (Eqn \ref{CDemod}) with a EWMA($\alpha$) filter $\mathbb{H}[ \cdot]$ given by (\ref{EWMA}). The expected values of $y_t$ and $\tilde y_t$ are
\begin{eqnarray*}
 E(y_t) = \frac{\cos{(\phi)}}{2} + b(y_t), \quad E(\tilde{y}_t) = \frac{\sin{(\phi)}}{2} + b(\tilde{y}_t)
\end{eqnarray*} where $b(y_t)$ and $b(\tilde{y}_t)$ are the biases  of $y_t$ and $\tilde{y}_t$, respectively. We have that 
\begin{eqnarray*}
b(y_t) &=&  \frac{(1-\alpha)\big(\cos{(2\omega t+\phi)} + \alpha\cos{(2\omega(t+1)+\phi})\big)}{2(1-2\alpha\cos{(2\omega)}+\alpha^2)} - \\ &&  \frac{\alpha^{t+1}}{2}\left(\cos{(\phi)} - \frac{(1-\alpha)(\cos{(\phi-2\omega)} - \alpha\cos{(\phi)})}{(1-2\alpha\cos{(2\omega)}+\alpha^2)}\right)   \\
b(\tilde{y}_t) &=& \frac{(1-\alpha)\big(\sin{(2\omega t+\phi)} - \alpha\sin{(2\omega(t+1)+\phi})\big)}{2(1-2\alpha\cos{(2\omega)}+\alpha^2)} -  \\ && \frac{\alpha^{t+1}}{2}\left(\sin{(\phi)} + \frac{(1-\alpha)(\sin{(\phi-2\omega)} - \alpha\sin{(\phi)})}{(1-2\alpha\cos{(2\omega)}+\alpha^2)}\right) \nonumber
\end{eqnarray*}
%
\end{thm}

\begin{proof} By direct calculation. \end{proof}

The results of Theorem \ref{T1} allow us to understand the how the design of the filter (\ref{EWMA}), in particular the choice of the weight $\alpha$ and the distance from the boundary $t$, affects the estimate of the instantaneous phase. From these expressions we can see two distinct components in the bias; one purely oscillatory component and a `boundary effect'  which goes to zero for large t, since $\alpha^{t+1} \rightarrow 0$. The complex demodulation algorithm introduces a $2f_0$ Hz component into the signal and  the oscillatory component represents the remains of this component after filtering. By increasing $\alpha$ in (\ref{EWMA}), the effectiveness of the filter is improved and the amplitude of the oscillatory component decreases, but this will lengthen the duration of the boundary effect. 

The reliable identification of phase shift events relies on the magnitude of the shift being greater than the bias; thus the selection of a filter  is a trade-off between the power to identify low magnitude shifts and the power to resolve shifts with small inter-shift-intervals (ISIs). For fixed $\alpha$, and large values of $t$, the magnitude of the bias is bounded
	\[b(y_t), b(\tilde{y}_t) \leq \frac{(1+\alpha)}{2(1-\alpha)},
\]
To estimate the overall bias in the estimator $\hat{\phi_t}$, we use a first order Taylor series expansion
	\[E(f(y_t,\tilde{y}_t)) \approx f(E(y_t), E(\tilde{y}_t)).
\]
This gives,
\begin{eqnarray*}
E\left(\tan^{-1}\left(\frac{\tilde{y_t}}{y_t}\right)\right) &\approx& \phi + \tan^{-1}\left(\frac{\cos{(\phi)}b(\tilde{y}_t) + \sin{(\phi)}b(y_t)}{1 + \cos{(\phi)}y(Y_t) + \sin{(\phi)}y(\tilde{Y}_t)}\right)
\end{eqnarray*}
For large $t$, we can bound the magnitude of the overall bias using 
\begin{eqnarray*}
b(\hat{\phi}_t) &\approx& \tan^{-1}\left(\frac{\cos{(\phi)}b(\tilde{y}_t) + \sin{(\phi)}b(y_t)}{1 + \cos{(\phi)}b(y_t) + \sin{(\phi)}b(\tilde{y}_t)}\right) \\
&\leq&  \tan^{-1}\left(\frac{(1+\alpha)}{\sqrt{2}(1-\alpha)-(1+\alpha) }\right), \quad \alpha > \frac{\sqrt{2}-1}{\sqrt{2}+1}
\end{eqnarray*}

Although the theoretical results from this section  are for a simple filter, they provide insight and intuition about the behaviour of more complicated filters, specifically the trade-off between size of persistent oscillatory bias and transient boundary induced bias. The existence of a non-zero bias term means that there will be a lower bound on the magnitude of shifts ($\Delta_{min}$) which can be reliably identified, and by reducing the noise levels we increase the range of identifiable events. Conversely, the boundary effect which occurs at the start of the recording will also be present at any phase shift event, and this will result in a minimum interval ($ISI_{min}$) such that two change point events can be accurately resolved. Decreasing the bandwidth in the complex demodulation algorithm (2$\delta$) will decrease overall levels of noise, but lengthen the duration of transient boundary effects in the instantaneous phase estimate; the bandwidth represents a trade-off between statistical power and temporal resolution. 

\subsection{Phase Shift Identification}
\label{phaseSynchrony}

The identification of phase shift events is considered from two different perspectives, as a post-hoc analysis which uses the entire recording, or real-time analysis which considers only  information available prior to the current time. The latter occurs when it is necessary to have an immediate response to a shift event, such as in statistical process control or BCI; in this situation it is very important to identify events with computational efficiency and high temporal resolution. 

To identify phase shift events, we first test the hypothesis of no phase shift events ($H_0: \phi_t = \phi_0$) against the alternative of a single phase shift event at time $t_0$ and with magnitude $\Delta$ ($H_a: \phi_t = \phi_0 + \Delta H(t - t_0)$). Two statistics are considered,
\begin{eqnarray}
S_1 &=& \max_{2 \leq t \leq {N-1}} s_1(t), \\
S_2 &=& \max_{2 \leq t \leq {N-1}} s_2(t),
\end{eqnarray}
where
\begin{eqnarray*}
s_1(t) &=& \left(\frac{N}{t(N-t)}\right)^{1/2} \sum_{i=1}^t \left(\hat \phi_i - \bar{\hat \phi}\right) \\
s_2(t) &=& \frac{|\hat \phi_{t+1} - \hat \phi_{t-1}|}{2},
\end{eqnarray*} and where $\hat \phi_t$ is the estimate of the instantaneous phase at time $t$ and $\bar{\hat \phi}$ the temporal average. The former is a cumulative summation (CUSUM) type statistic which is favoured in traditional change-point analysis (CPA), while the latter corresponds to high temporal resolution phase derivative (PD) method which is currently applied to EEG recordings. 

The critical value of the hypothesis test ($\Phi_\alpha$), will be calculated using bootstrapping and approximation techniques (see \S \ref{para} and \S \ref{nonpara}). If the null hypothesis is rejected, a change point time $\hat{t}_0$ is estimated such that
	\[s_{i}(\hat{t}_0) = \max_t s_{i}(t), \quad i=1,2.
\]
Once $\hat{t}_0$ has been identified, the estimator can be iteratively applied to both halves of the sample to search for additional phase shift events. This process continues until either no more shifts are found, or there are less than $N_{min}$ points in a signal. Note that the existence of a phase shift event causes an additional `boundary effect' at the discontinuity of the instantaneous phase. Such a boundary effect obscures the phase dynamics around the shift event. To ensure additional spurious phase shift events are not detected due to this boundary effect, an interval $(t_L, t_U)$ about the estimated change-point must be excluded from further analyses. 

For the CUSUM estimator, we use bootstrapped datasets to estimate the parameters $N_{min}$, $t_L$ and $t_U$. Bootstrapped signals with no phase shift events are used to find the minimum value $N_{min}$ which achieves the desired false positive rate, 
	\[P(S_1(x_{1:N_{min}}) > \Phi_{\alpha}) \leq \alpha
\]
Additionally, bootstrapped signals with a single phase shift event are used to estimate a value $ISI_{min}(\alpha)$, such that if 
	\[t_L = \hat{t}_0 - ISI_{min}(\alpha), \quad t_U = \hat{t}_0 + ISI_{min}(\alpha),
\]
then the lower and upper subsamples $x_{1:t_L}$ and $x_{t_U:N}$ have the desired sampling distribution, that is, 
	\[P(S_1(x_{1:t_L}) > \Phi_{\alpha}) \leq \alpha, \quad P(S_1(x_{t_U:N}) > \Phi_{\alpha}) \leq \alpha.
\]

For the PD estimator, since this is an instantaneous estimator, we need not be so conservative in the exclusion of data around the change-point. In fact, it is only required to remove any points around $\hat{t}$ which are greater than the critical value, 
	\[t_L = \max_{t<\hat{t}_0}\{t | s_2(t) < \Phi_{\alpha}\}, \quad t_U =  \min_{t>\hat{t}_0}\{t | s_2(t) < \Phi_{\alpha}\}
\]

\subsection{Parametric Bootstrapping}
\label{para}

To learn about the sampling distribution of estimators for a more complicated filter than considered in \S \ref{theory}, a parametric bootstrapping procedure is used,  \citep{Efron1979}. Bootstrapping techniques are well developed for CPA in independent observations~\citep{Antoch2001}. To overcome the problem of dependencies, block bootstrapping techniques for CPA have been developed which preserve the temporal correlations in the bootstrapped datasets~\citep{Kirch2007}. 

Here we consider a simple 9 Hz ($f_0=9$) oscillator, sampled at T=250 Hz with constant unit amplitude. To account for the boundary effects in phase estimation, the first $N_{burn}$ samples are removed from the analysis,
	\[x_t = \sin\left(\frac{2 \pi f_0 t}{T} + \phi_t\right), t=N_{burn} \ldots N.
\]
In this parametric bootstrap,  data are generated with i.i.d.~additive noise ($\epsilon_t \sim N(0,1)$). Both the oscillator and noise term are normalized to unit power, and the signal-to-noise ratio (SNR) of the simulated observable is set by a weight parameter $r$,
	\[x^{(b)}_t = \frac{rx_t}{||x_t||} + \frac{(1-r)\epsilon_t^{(b)}}{||\epsilon_t^{(b)}||}, \quad b=1..B.
\]The relationship between $r$ and the SNR is given by 
	$SNR = 10\log_{10}\frac{r^2}{(1-r)^2}.$

For large values of N, the sampling distribution of the maximum value in a sequence of random variables converges to the generalized extreme value (GEV) family of distributions; this result has been proven for a broad class of dependent random variables, and does not require a stationarity assumption~\citep{Galambos1972}. 

In Supplement A we present a simulation study which explores the behaviour of $S_1$ and $S_2$ in a simple parametric application. The results confirm the intuitive difference between a cumulative and instantaneous estimator. The $S_1$ (cumulative) estimator has greater power to identify phase shift events, especially with high noise or low effect size. Conversely, the $S_2$ (point-wise) estimator is better at resolving multiple shift events, it is able to identify shifts with a lower ISI. 

\subsection{Non-Parametric Methods}
\label{nonpara}

The above parametric procedures, where the noise is assumed i.i.d., provides a clear intuition about the range of shift magnitudes and SNRs where a shift can be reliably identified; however, such methods are not appropriate for all situations. Neither the R\"{o}ssler attractor nor EEG recordings can be accurately described by an ideal oscillator with i.i.d.~noise. More often, signals from complex systems have oscillations which occur around, but not exactly on a fixed frequency. Other factors which can affect the detection of shift events are time dependent noise, or the existence of unrelated systematic features of the system such as oscillators at different frequencies. 

To obtain a nonparametric estimate of the sampling distributions, it is necessary to have some measure of the temporal dependence which remains after the filtering process. Here we use the first zero crossing of instantaneous phase autocorrelation function ($\tau$), although other methods may also be appropriate such as using the local minimum of the time lagged mutual information function~\citep{Fraser1986}. The existence of phase shift events will artificially inflate the value of $\tau$, so ACF functions should ideally be estimated from data which does not contain shift events. Additionally, due to the non-stationary nature of the signal, estimates of $\tau$ should be combined from several different points in the signal. Here we take the arithmetic average of the $\tau$ values, though the median could also be used. \\

\subsubsection{Block Bootstrapping}

To estimate the sampling distribution of the CUSUM estimator ($S_1$), we employ the nonparametric block-permutation bootstrap technique. There are several different blocking techniques, notably the moving block~\citep{Kunsch1989} and circular bootstrap~\citep{Politis1992}. Here we use a non-overlapping block method, \citet{Kirch2007}, which is well studied in the change-point paradigm, the observed time series is partitioned into K intervals of length L, 
	\[x(k) = x_{1+L(k-1)}:x_{Lk}, \quad k = 1..K.
\]
For a randomly generated permutation $\pi$ of (1..K), the surrogate dataset is 
	\[x^{(b)} = [x(\pi(1)), x(\pi(2)), ..., x(\pi(K))].
\]
Since correlations with-in each block remain unchanged, the surrogate datasets mimic the correlation structure of the underlying process. The value of L is chosen such that each window fully captures the autocorrelation structure of the signal, here we take a value of $L=2\tau$. 

\subsubsection{Threshold Method} 

The method of block bootstrapping, is not applicable for the $S_2$ estimator. Creating a block permutation of the original data does not have the desired effect on the estimator, all with-in block values of $s_2$ remain unchanged; this is because the nature of the estimator is instantaneous rather than  cumulative. 

We suggest an approximation to the distribution of $S_2$ which takes into account the variance and autocorrelation of the instantaneous phase. The phase differences are first centred and normalized to create a standardized sequence of N variables. To account for the autocorrelations, we treat each block as an independent observation and the distribution of the maximum, $S_2$, is approximated as the maximum of $K^*=2\left\lfloor{N/\tau}\right\rfloor$ independent random variables. Critical values are based on the maximum of independent Normal random variables, although for large values of $K^*$ these will converge to the GEV distribution.  

In the case of multiple phase shift events, rather than simply applying the method recursively, we suggest a pooled estimate of the variance which includes all parts of the signal which have not be identified as a phase shift event. 
\begin{eqnarray}
\hat{\sigma}_{pool}^2 &=& \frac{\sum_{j=1}^k (N_j-1)\hat{\sigma}_j^2}{\sum_{j=1}^k{(N_j-1)}} \\
\hat{\sigma}_j^2 &=& \frac{1}{N_j-1}\sum_{i=t_L(j)}^{t_U(j)} \left(s_1(j) - \bar{s_1}(t_L:t_U)\right)^2 \\
N_i &=& t_U(i)-t_L(i) + 1 
\end{eqnarray}
Thus the algorithm is: 
\begin{enumerate}
\item Estimate the standard deviation of the phase locked values ($\hat \sigma^2_{pool}$)
\item Set the threshold to $\Phi_{\alpha} = \hat \sigma_{pool} z_{\left(\alpha^{K^*_i}\right)}$
\item Update the set of boundaries ($t_L$, $t_U$)
\item Test for phase shift events
\begin{itemize}
	\item If any new shift events are identified, update the set of phase shift events ($t_L$, $t^*$, $t_U$) and then go to step 1
	\item Otherwise stop
\end{itemize}

\end{enumerate}

As phase shift events are removed from the estimate, $\sigma_{pool}$ will decrease, causing the critical value $\Phi_{\alpha}$ to be monotonically decreasing. For this reason, the boundaries ($t_L$ and $t_U$) must be updated at each iteration to account for the lowered threshold. Occasionally, lowering the threshold will cause two shift events to occur consecutively, in this situation it is not clear if there actually two events, so they should be merged into a single event rather than risk falsely declaring a spurious shift. 

In addition to the increased temporal resolution, another benefit of this approximation is its computational efficiency compared to the bootstrapping techniques presented. The current algorithm is described as a post-hoc analysis of the entire signal; however, it could be adapted into an online implementation.

\section{Simple Oscillators}
\label{Osci}

In this section  the accuracy of the suggested methods is evaluated against the parametric alternative in the case of simple oscillators with i.i.d.~noise. For each application, a fourth order low-pass Butterworth filter is used to estimate the instantaneous phase because it has a maximally flat frequency response, i.e. it reduces the amount of ripple in the pass-band. To assess the methods, we use simulated oscillators which have a SNR of 0 dB and include $M=20$ phase shift events. The shift magnitudes and ISIs are randomly drawn to obtain a robust comparison,
\begin{eqnarray*}
\phi_t &=& \sum_{i=0}^M \left(\sum_{j=0}^i \Delta_j \right)H(t - t_i) \\
\Delta_j &\sim& Unif\left((-\pi, -\Delta_{min}]\cup[\Delta_{min}, \pi]\right) \\
t_i &=& t_{i-1} + ISI_{min} + Exp(ISI_{min})
\end{eqnarray*}
Each estimator is evaluated based on the rate of true positives (TP), true negatives (TN), false positives (FP) and false negatives (FN). The ability of each method to correctly detect change-point events are displayed for a range of significance level using receiver operating characteristic (ROC) curves, which plot TP against FP. The accuracy (ACC = (TP+TN) / (TP+FP+TN+FN) ) provides a measure of goodness at specific significance levels, while the area under an ROC curve (AUROC) provides an overall measure of quality. The accuracy measure gives equal weight to TP and TN events, this is not necessarily optimal for every situation, other weightings may be considered depending on the application. 

The resulting ROC curves from the change point analysis of twenty simulated oscillators (for a total of 400 shift events) are shown in Figure \ref{Figure4}. Both the CUSUM and PD estimators have high power to identify phase shift events, while controlling  the false positive rates. There is a cross-over of the ROC curves for the nonparametric estimators, which implies neither method is uniformly better than the other; the PD method provides the maximum accuracy, but the CUSUM method can provide increased power at the cost of a slightly higher FP rate. The parametric CUSUM method works best overall, and significantly better than the nonparametric CUSUM, while both the parametric and nonparametric PD methods give similar results. Table \ref{Table1} shows the maximum ACC and AUROC values for each method, both nonparametric methods perform well, though the PD estimator is marginally better. 

\begin{figure}[!ht]
\begin{center}
\includegraphics[width = 4in]{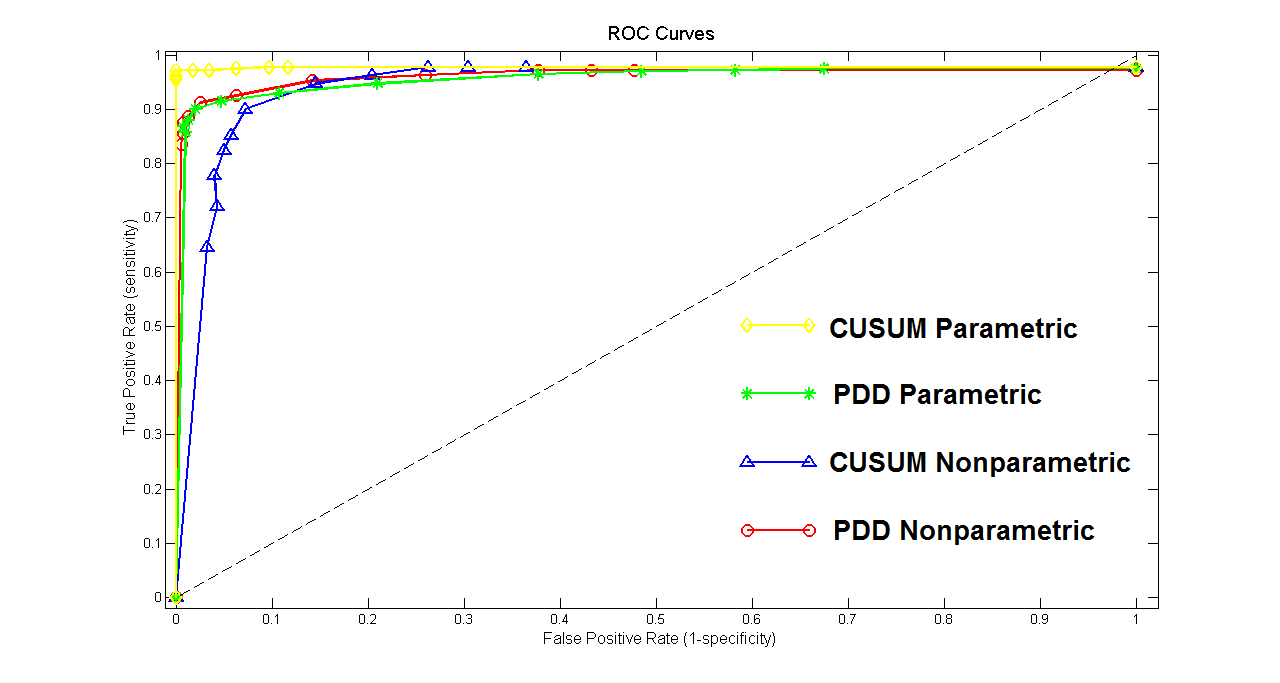}
\caption{The ROC curves for both the parametric and nonparametric versions of the CUSUM and PD estimators. The parametric CUSUM estimator (yellow) performs the best, dominating each other curve. The parametric (green) and nonparametric (red) PD estimators perform similarly, they are both better than the nonparametric CUSUM (blue) for the low FP rates but then there is a cross-over and the nonparametric CUSUM performs better for high FP ranges.}
\label{Figure4}
\end{center}
\end{figure}

\begin{table}[!ht]
\begin{tabular}{ccc}
Method & mACC & AUROC \\
\hline
CUSUM Parametric $S_1$ & 0.9863 ($\alpha=0.1$) & 0.9772 \\
PD Parametric $S_2$ & 0.9400 ($\alpha=0.01$) & 0.9555 \\
CUSUM Nonparametric $S_1$ & 0.9137 ($\alpha=0.1$)& 0.9438 \\
PD Nonparametric $S_2$ & 0.9438 ($\alpha=0.03$) & 0.9610
\end{tabular}
\caption{A table of the maximum accuracy (mACC) and area under the curve (AUROC) for the simple oscillators. Both nonparametric measures provide comparable results, although the instantaneous estimator ($S_2$) has both the highest mACC and AUROC. }
\label{Table1}
\end{table}

\section{R\"{o}ssler Attractor}
\label{Ross}

In this section, we apply the proposed methods to a system of coupled R\"{o}ssler attractors. This application represents an intermediate level of difficulty between the i.i.d. oscillators and an unsupervised EEG application. The chaotic dynamics are not strictly periodic and significant power leaks into nearby frequency bands, this is similar to the spectral properties of EEG signals; however, unlike an EEG application the true shift events in the R\"{o}ssler attractor are still available to calibrate methods.  

The R\"{o}ssler system is a set of three ordinary differential equations~\citep{Rossler1976}, which is linear except for a single bi-linear term. The attractor primarily rotates around the origin in the x-y axis, while showing spontaneous bursts in the z-direction. Sample trajectories generated with random initial conditions are shown in Figure \ref{Figure5}; a burn-in period of 30 seconds is thrown away to remove transient activity. 

\begin{figure}[!ht]
\begin{center}
\includegraphics[width=4in]{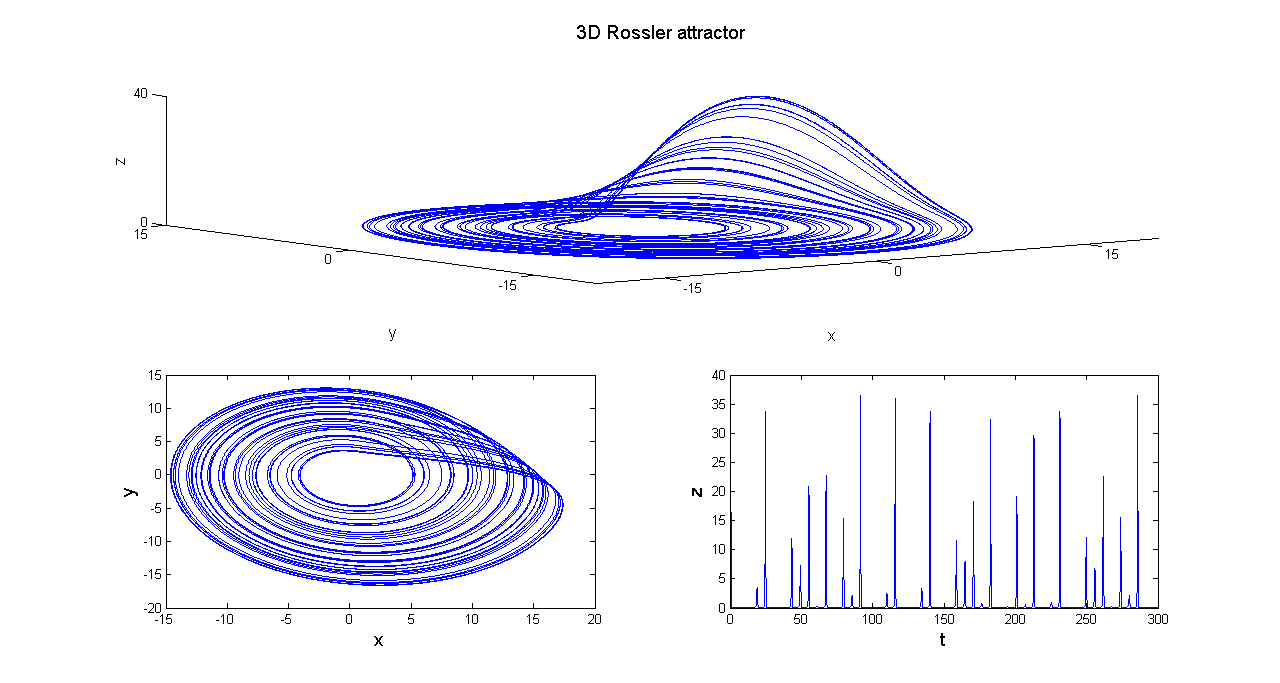}
\caption{Sample trajectory of a R\"{o}ssler attractors.}
\label{Figure5}
\end{center}
\end{figure}

Here we explore the two coupled R\"{o}ssler attractors, 
\begin{eqnarray*}
\dot{x}_{1/2} &=& \omega_{1/2}(-y_{1/2}-z_{1/2}) + C(x_{1/2} - x_{2/1})   \\
\dot{y}_{1/2} &=& \omega_{1/2}(x_{1/2} + ay_{1/2})  \\
\dot{z}_{1/2} &=& \omega_{1/2}(b + z_{1/2}(x_{1/2}-c)) 
\end{eqnarray*}
with $a=0.15$, $b=0.2$, $c=10$. The frequency parameters were $\omega_1 = 2\pi f_0 + \delta\omega$ and $\omega_2 = 2\pi f_0 - \delta\omega$, where  the attractors have average frequency of $f_0=9Hz$ and a frequency mismatch of $\delta\omega=0.675$. For this set of frequency parameters, the phase shift dynamics observed in ~\citet{Rosenblum1996} are recreated by setting the coupling parameter to $C = 0.12$. Trajectories were generated at a rate of 10 kHz and then down-sampled to 250 Hz. Estimated power spectral densities are shown in Figure \ref{Figure6}, the attractor shows a distinct peak at approximately $9.25$ Hz. 

To estimate the autocorrelation which remains in the instantaneous phase, we increase the coupling value to C=0.5 to generate signals with strong synchronization and no shift events. The ACF of the instantaneous phase difference for the strongly coupled oscillators is shown in Figure \ref{Figure6}. This first zero crossing of the ACF function occurs at $\tau=1183$, or approximately 4.5 seconds. 
\begin{figure}[!ht]
\begin{center}
\includegraphics[width=4.5in]{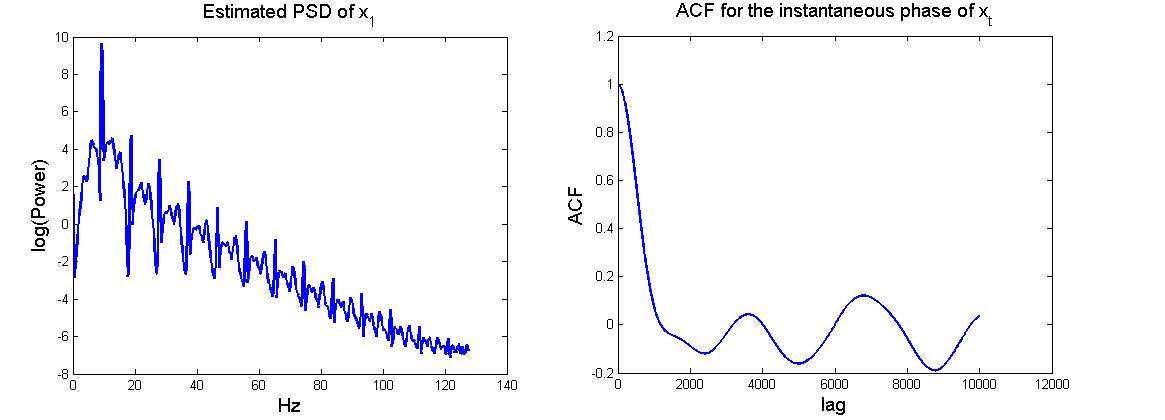}
\caption{Left: Estimated log PSD of $x_1$ for the coupled attractors (densities based on $x_2$, $y_{1/2}$ were nearly identical). There is a distinct peak at approximately $9.25$ Hz. The signal is not strictly periodic and the spectral components `leak' into other bands. Notably, there is a periodic beat frequency due to the frequency mismatch ($d\omega$) of the oscillators. Right: Estimated ACF function of the instantaneous phase of $x_t$ during a period of phase synchronization (no shift events). The first zero crossing of the ACF can be used to determine the appropriate value of L for the block bootstrapping algorithm. Here the crossing occurs at a lag of 1183, or approximately 4.5 seconds.}
\label{Figure6}
\end{center}
\end{figure}
Trajectories primarily oscillate in the x-y plane, so there is a natural definition for the phase of this attractor (in the Poincare map sense), which will be take as the true phase of the system,
\begin{equation}
\phi_t = tan^{-1}\left(\frac{y_t}{x_t}\right).
\end{equation}
Instantaneous phase variables can also be defined in the Hilbert transform sense, using observables of the system. The complex demodulation algorithm is applied with a centre frequency of $\omega = 9.25$ Hz (see Figure \ref{Figure6}) and a bandwidth of $\Delta\omega = 0.15$ Hz; resulting phase differences between the two coupled oscillators are shown in Figure \ref{Figure7}, for both the true phase and the Hilbert phase with observables $h(x_t,y_t,z_t) = x_t$ and $h(x_t,y_t,z_t) = y_t$.

\begin{figure}[!ht]
\begin{center}
\includegraphics[width=4in]{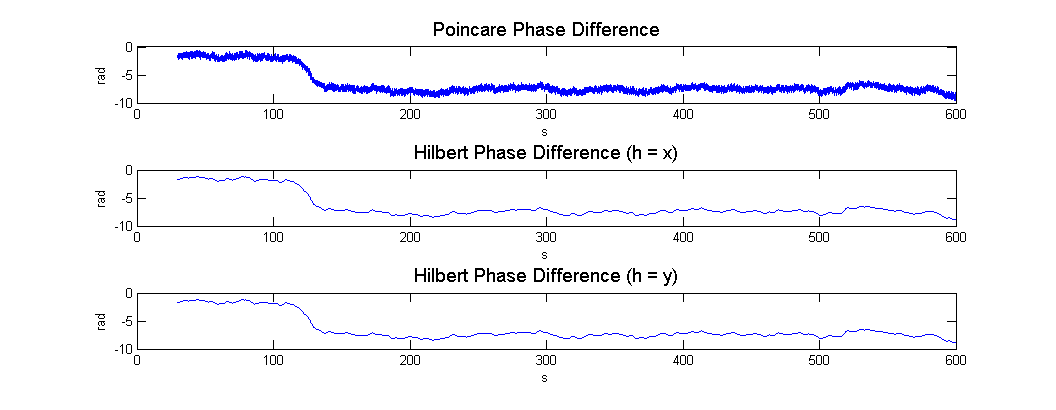}
\caption{Coupled R\"{o}ssler attractors with $a=0.15$, $b=0.2$, $c=10$, $C=0.12$, $w=2\pi 9$, $dw=0.0675$. Phase difference plots for three definitions of phase; top: Poincare phase, middle: Hilbert transform of $x_t$, bottom: Hilbert phase of $y_t$. In all three definitions, there is a clear locking / shifting dynamic, with a phase shift at approximately $s=130$ seconds.}
\label{Figure7}
\end{center}
\end{figure}

\subsection{R\"{o}ssler Results}

To assess the ability of each method to identify phase shift events in the weakly coupled R\"{o}ssler attractors, we generated fifty datasets with random initial conditions and a length of 10 minutes. Phase shift events were manually marked using the Poincare definition of phase. There were 139 total phase shift events for an average of 2.78 per dataset. 

Using the observable $h=x_t$, shift events are estimated using both nonparametric methods. An ROC curve of the results (see Figure \ref{Figure8}) shows that the CUSUM estimator outperforms the PD estimator. There are several potential reasons why the CUSUM estimator outperforms the PD in this situation, in contrast to independent oscillators; (1) longer average ISI's, (2) uni-directed phase shifts, and (3) temporally correlated `noise'. The PD estimator eventually crossed the CUSUM estimator in the ROC curve, this is due to the few shift events which occur with an ISI that is too small to be resolved by the CUSUM estimator. The maximum accuracy (mACC) and area under the curve (AUROC) for these methods are shown in Table \ref{Table2}, here we see that both estimators perform well but the CUSUM estimator performs best.

\begin{figure}[!ht]
\begin{center}
\includegraphics[width=4in]{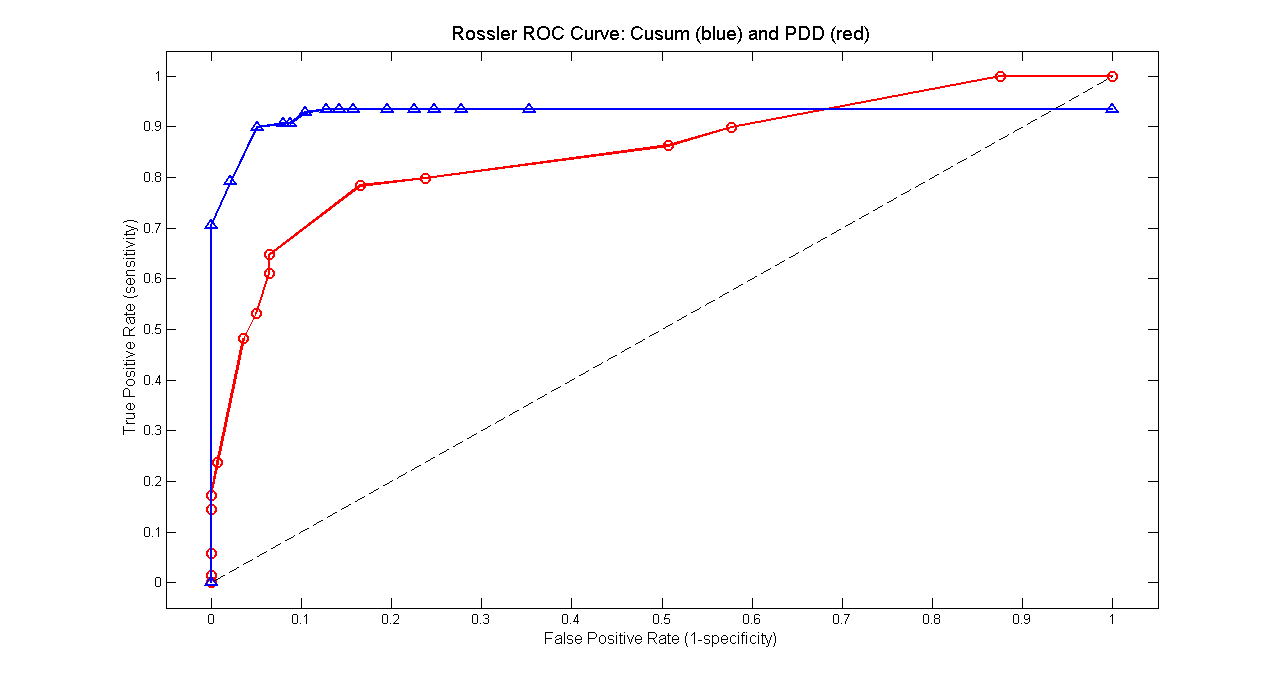}
\caption{Plot of the ROC curve of the nonparametric CUSUM (blue) and PD (red) phase shift estimators, applied to simulated data from coupled R\"{o}ssler attractors. For this application, the CUSUM estimator outperforms the PD estimator in both mACC and AUROC. The intervals between shift events from the attractors are long enough that the CUSUM method identifies over 90\% of them. The PD is able to identify all shift events, but does so with high false positive rate, due to the `temporally correlated noise' from the attractor dynamics.}
\label{Figure8}
\end{center}
\end{figure}

\begin{table}[!ht]
\begin{tabular}{ccc}
Method & mACC & AUROC \\
\hline
Nonparametric CUSUM $S_1$ & 0.8225 ($\alpha=0.05$)& 0.9238 \\
Nonparametric PD $S_2$ & 0.6187 ($\alpha=0.1$) & 0.8504
\end{tabular}
\caption{A table of the maximum accuracy (mACC) and area under the curve (AUROC) measures of the ROC plot for coupled R\"{o}ssler attractors. For this application the CUSUM estimator outperforms the PD estimator.}
\label{Table2}
\end{table}

To investigate potential power law behaviour in the distribution of ISIs, we generate one thousand twenty-minute datasets with random initial conditions. Applying the CUSUM estimator with $\alpha=0.05$ found a total of 5661 phase shift events, with an average ISI of $\mu=2.72$ minutes ($\sigma=1.66$). Histograms of the observed ISIs are shown in Figure \ref{Figure9}, on standard and log-log scales; the log-log histogram clearly shows an asymptotic power-law behaviour with an estimated slope of $q=-3.92$.

\begin{figure}
\begin{center}
\includegraphics[width=4in, height=2in]{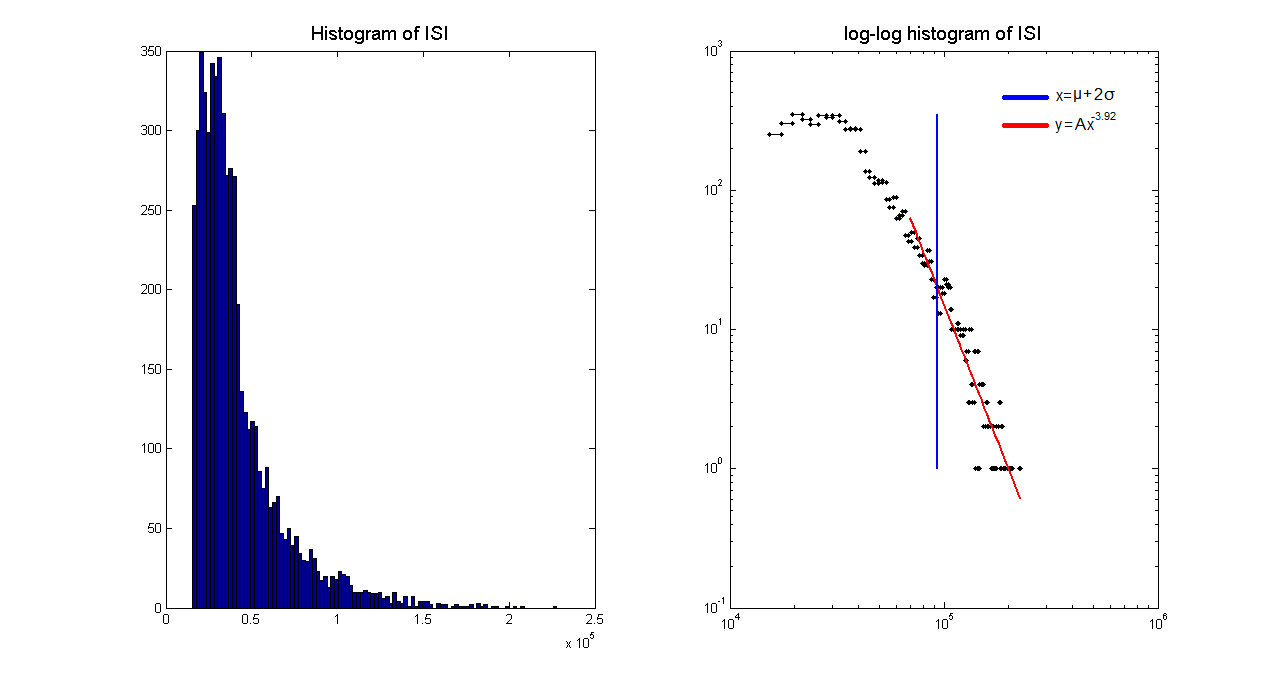}
\caption{Left: Histogram of the observed ISIs in the weakly coupled R\"{o}ssler attractors, as estimated by the CUSUM estimator. Right: corresponding log-log scale histogram, there is a linear relationship in the tail of the distribution, with a scaling exponent of 3.92.}
\label{Figure9}
\end{center}
\end{figure}

\section{EEG Phase Shift}
\label{EEG}

This section applies the methodology to EEG recordings. These signals are comprised of many different components, and as in the case of the R\"{o}ssler attractors, there is power in multiple spectral bands. In this real world problem there is no information on the truth of phase shift events to calibrate the algorithms, thus the utility of the methods is assessed by their ability generate neurologically plausible results. Relationships between shift events and external visual stimuli are explored, with emphasis on the scalp regions which are associated with such a visual task. Additionally, the distribution of ISI's is explored, looking for evidence of the asymptotic power-law distribution which are hypothesised to exist~\citep{Sporns2004}. 
 
In order to assess the proposed methods, EEG recordings were obtained from 18 participants (9 male, 9 female, aged 12-14 years) during a visual task, specifically a modified Erikson Flanker task~\citep{Eriksen1974}. The task was to discriminate two stimuli by pressing the corresponding button for each. Stimuli were presented for 200 ms followed by a variable inter-trial interval (ITI) of 800 to 900 ms. The task took approximately 15 minutes to complete. EEG recordings were obtained from 121 scalp sites (EGI, Eugene, OR) at a sampling rate of 500 Hz. The recordings were reduced to a set of 16 standard sites (see Figure \ref{Figure11}) representing regions of the left hemisphere (Fp1, F7, F3, C3, P3, T3, T5, O1) and the right hemisphere (Fp2, F8, F4, C4, P4, T4, T6, O2). Impedances were maintained below 30 k$\Omega$ throughout recording. Data were re-referenced offline to the average of all sites and corrected for eye movements using the Gratton and Coles procedure~\citep{Gratton1983}. An automated artifact rejection procedure was used in addition to manual examinations of the data. 

The beta band (13-30Hz) of EEG recordings is often associated with sensorimotor activity~\citep{Pfurtscheller1981}, such as in the flanker task in our data.  A recent attempt at a unifying hypothesis of the functional role of beta band oscillations suggests that it is responsible for maintenance of sensorimotor or cognitive state~\citep{Engel2010}. Spectral power analyses of the beta-band has been previously employed for classification in BCI applications~\cite{Bai2008}. Here we focus specifically on the values $\omega=16.5$ and $\Delta\omega=3.5$, corresponding to the (13-20Hz) lower beta band. 

To estimate the length of the autocorrelations in the EEG application, we first divide the recordings into 4 second segments. For each segment and each pair of channels, we estimate first zero-crossing ($\tau$) in the ACF of the wrapped instantaneous phase. The value of K is calculated based on the average of all the estimated values of $\tau$, this results in a value of $K=85$. 

For each pair of signals, we apply the PD identification algorithm to the instantaneous phase difference of the pairs, to identify spontaneous desynchronizations. In total there were $4\,690\,214$ phase shift events across all 18 participants and $16\times 15=120$ pairs of electrodes, with an average ISI of 262 ms (standard deviation 226 ms) or 3.8 shifts per second. We also investigate power-law behaviour in the distribution of ISI of beta band phase shift events. Standard and log-log scale histograms of the distribution of ISIs are shown in Figure \ref{Figure10}. There appears to be asymptotic power-law behaviour in the tail of the distribution with a scaling exponent of $q=6.26$ 

\begin{figure}
\begin{center}
\includegraphics[width=4in, height=2in]{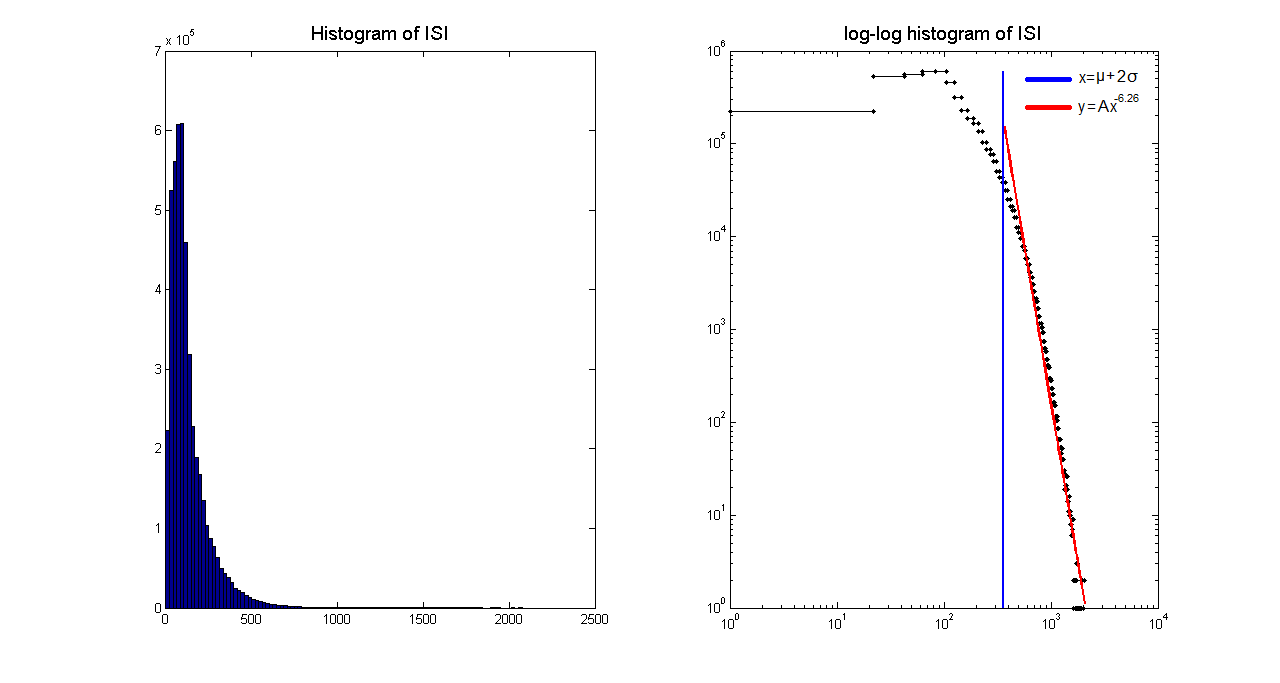}
\caption{Left: Histogram of the observed ISIs in the beta band activity of EEG recordings during a visual vigilance task, as estimated by the PD estimator. Right: The corresponding log-log scale histogram, there is a linear relationship in the tail of the distribution, with a scaling exponent of 6.26.}
\label{Figure10}
\end{center}
\end{figure}

We further consider the relationship between the occurrence of phase shift events and the task stimuli. For each pair of signals (i,j), and each shift event (k), we record the amount of time since the most recent stimulus event. If there is no relationship between the shift and stimuli, then we expect that the times will be uniformly distributed. To test this hypothesis, we group the variable into 10 equal sized bins, between 0-500 ms and apply a $\chi^2$ test for uniformity. Results from this analysis are summarized in Figure \ref{Figure11}; there are three pairs of sites (Fp1-F7, T3-O1, T6-O1) which are significant at the $\alpha=0.05$ level, including a Bonferroni correction $(p < \alpha/120)$ and an additional six pairs (Fp1-T3, Fp1-O1, Fp2-O1, T3-T4, T5-O1, F8-P3) which are significant with a less conservative correction ($p < 0.05/30$). As expected for our visual attention task, many of these pairs involve the occipital (O1,O2) sites over the visual cortex and prefrontal sites (Fp1, Fp2), which are associated with attention.  

\begin{figure}
\begin{center}
\includegraphics[width=2.8in]{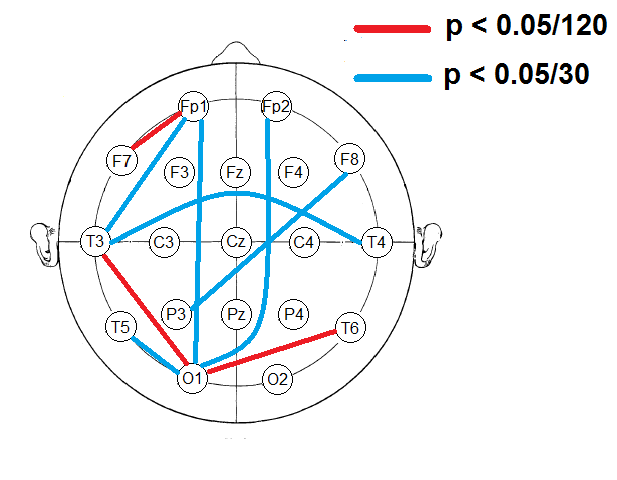}
\caption{The standard 10-20 system introduced in~\cite{Jasper1958}. Results of the tests for uniformity of phase shift events between visual stimuli. Red lines represent significance with a full Bonferroni corrected p-value ($p < 0.05/120$) while blue lines represent significance at a less conservative level ($p < 0.05/30$). The region most commonly related to the stimuli is the left occipital region O1, directly above the visual cortex. Additionally, the frontal sites Fp1 and Fp2, commonly associated with attention, are also related to the stimuli.}
\label{Figure11}
\end{center}
\end{figure}

\section*{Discussion}

In applications such as BCI, it is desirable to not only identify events with high temporal resolution, but also in real-time so that feedback can be provided immediately. The instantaneous nature of the PD estimator, as well as its computational efficiency, make it easily modified to perform in real-time. Conversely, there is no obvious analog for the CUSUM  estimator, which requires access to the entire signal, as well as time consuming bootstrap procedures, and as such is better suited to post-hoc analysis. 

Regarding the stability of the stability of the proposed methods, both methods are stable in the identification of a single phase shift event. In the case of multiple shift events, it is possible that small changes in parameters may have a large effect on results of the CUSUM estimator. If the first change-point is not identified, then all other potential shifts are not identified. This is not the case in with the PD estimator, where the instantaneous nature causes the identification of a change-point to be independent from other change points, resulting in a more stable algorithm.  

There are often many sources of noise in measuring observable time series from complex systems; systematic effects, additional oscillating components, temporally correlated noise, spatially correlated noise (i.e. source mixing in EEG due to volume conduction), frequency misspecification or narrow-band signals (non-fixed frequency). An investigation into the effect of such features in controlled environment may provide insight as to how to formulate a more robust solution to the shift identification problem.

\section*{Acknowledgements} We would like to thank Dr. Sid Segalowitz and the Brock University Laboratory of Cognitive and Affective Neuroscience for providing us with the EEG data used in this study.

\begin{supplement}[id=suppA]
  \sname{Supplement A}
  \stitle{Simulation Exercise}
  \slink[doi]{COMPLETED BY THE TYPESETTER}
  \sdatatype{.pdf}
  \sdescription{We present a simulation study exploring the behaviour of $S_1$ and $S_2$ in a parametric application.}

\end{supplement}

\bibliography{PhaseShift}

\begin{thebibliography}{35}

\bibitem[\protect\citeauthoryear{Antoch and Hu\v{s}kov\'{a}}{2001}]{Antoch2001}
\begin{barticle}[author]
\bauthor{\bsnm{Antoch},~\bfnm{J.}\binits{J.}} \AND
  \bauthor{\bsnm{Hu\v{s}kov\'{a}},~\bfnm{M.}\binits{M.}}
(\byear{2001}).
\btitle{Permutation tests in change point analysis}.
\bjournal{Statistics and Probability Letters}
\bvolume{53}
\bpages{37-46}.
\end{barticle}
\endbibitem

\bibitem[\protect\citeauthoryear{Bai et~al.}{2008}]{Bai2008}
\begin{barticle}[author]
\bauthor{\bsnm{Bai},~\bfnm{O.}\binits{O.}},
  \bauthor{\bsnm{Lin},~\bfnm{P.}\binits{P.}},
  \bauthor{\bsnm{Vorbach},~\bfnm{S.}\binits{S.}},
  \bauthor{\bsnm{Floeter},~\bfnm{M.~K.}\binits{M.~K.}},
  \bauthor{\bsnm{Hattori},~\bfnm{N.}\binits{N.}} \AND
  \bauthor{\bsnm{Hallett},~\bfnm{M.}\binits{M.}}
(\byear{2008}).
\btitle{A high performance sensorimotor beta rhythm-based brain-computer
  interface associated with human natural motor behavior}.
\bjournal{J. Neural Eng.}
\bvolume{5}
\bpages{24-35}.
\end{barticle}
\endbibitem

\bibitem[\protect\citeauthoryear{Bak, Tang and Wisenfeld}{1987}]{Bak1987}
\begin{barticle}[author]
\bauthor{\bsnm{Bak},~\bfnm{P.}\binits{P.}},
  \bauthor{\bsnm{Tang},~\bfnm{C.}\binits{C.}} \AND
  \bauthor{\bsnm{Wisenfeld},~\bfnm{K.}\binits{K.}}
(\byear{1987}).
\btitle{Self Organized Criticality - An explanation of 1/f noise}.
\bjournal{Phys Rev Lett}
\bvolume{59}
\bpages{381-384}.
\end{barticle}
\endbibitem

\bibitem[\protect\citeauthoryear{Balocchi et~al.}{2004}]{Balocchi2004}
\begin{barticle}[author]
\bauthor{\bsnm{Balocchi},~\bfnm{R.}\binits{R.}},
  \bauthor{\bsnm{Menicucci},~\bfnm{D.}\binits{D.}},
  \bauthor{\bsnm{Santarcangelo},~\bfnm{E.}\binits{E.}},
  \bauthor{\bsnm{Sebastiani},~\bfnm{L.}\binits{L.}},
  \bauthor{\bsnm{Gemignani},~\bfnm{A.}\binits{A.}},
  \bauthor{\bsnm{Ghelarducci},~\bfnm{B.}\binits{B.}} \AND
  \bauthor{\bsnm{Varanini},~\bfnm{M.}\binits{M.}}
(\byear{2004}).
\btitle{Deriving the respiratory sinus arrhythmia from the heartbeat time
  series using empirical mode decomposition}.
\bjournal{Chaos, Solitons \& Fractals}
\bvolume{20}
\bpages{171-177}.
\end{barticle}
\endbibitem

\bibitem[\protect\citeauthoryear{Bingham, Godfrey and
  Tukey}{1967}]{Bingham1967}
\begin{barticle}[author]
\bauthor{\bsnm{Bingham},~\bfnm{C.}\binits{C.}},
  \bauthor{\bsnm{Godfrey},~\bfnm{M.~D.}\binits{M.~D.}} \AND
  \bauthor{\bsnm{Tukey},~\bfnm{J.~W.}\binits{J.~W.}}
(\byear{1967}).
\btitle{Modern Techniques of Power Spectrum Estimation}.
\bjournal{IEEE Trans on Audio and Electroacoustics}
\bvolume{AU-15}
\bpages{56-66}.
\end{barticle}
\endbibitem

\bibitem[\protect\citeauthoryear{Boccaletti et~al.}{2002}]{Boccaletti2002}
\begin{barticle}[author]
\bauthor{\bsnm{Boccaletti},~\bfnm{S.}\binits{S.}},
  \bauthor{\bsnm{Kurths},~\bfnm{J.}\binits{J.}},
  \bauthor{\bsnm{Osipov},~\bfnm{G.}\binits{G.}},
  \bauthor{\bsnm{Valladares},~\bfnm{D.~L.}\binits{D.~L.}} \AND
  \bauthor{\bsnm{Zhou},~\bfnm{C.~S.}\binits{C.~S.}}
(\byear{2002}).
\btitle{The Synchronization of chaotic systems}.
\bjournal{Phys Rep}
\bvolume{366}
\bpages{1-101}.
\end{barticle}
\endbibitem

\bibitem[\protect\citeauthoryear{Braitenberg}{1985}]{Braitenberg1985}
\begin{barticle}[author]
\bauthor{\bsnm{Braitenberg},~\bfnm{V.}\binits{V.}}
(\byear{1985}).
\btitle{Charting the Visual Cortex}.
\bjournal{Cerebral Cortex}
\bvolume{3}
\bpages{379-414}.
\end{barticle}
\endbibitem

\bibitem[\protect\citeauthoryear{Chavez et~al.}{2006}]{Chavez2006}
\begin{barticle}[author]
\bauthor{\bsnm{Chavez},~\bfnm{M.}\binits{M.}},
  \bauthor{\bsnm{Besserve},~\bfnm{M.}\binits{M.}},
  \bauthor{\bsnm{Adam},~\bfnm{C.}\binits{C.}} \AND
  \bauthor{\bsnm{Martinerie},~\bfnm{J.}\binits{J.}}
(\byear{2006}).
\btitle{Towards a Proper Estimation of phase sychronization from time series}.
\bjournal{J Neurosci Methods}
\bvolume{154}
\bpages{149-160}.
\end{barticle}
\endbibitem

\bibitem[\protect\citeauthoryear{Efron}{1979}]{Efron1979}
\begin{barticle}[author]
\bauthor{\bsnm{Efron},~\bfnm{B.}\binits{B.}}
(\byear{1979}).
\btitle{Bootstrap methods: Another look at the jackknife}.
\bjournal{Annals of Statistics}
\bvolume{7}
\bpages{1-26}.
\end{barticle}
\endbibitem

\bibitem[\protect\citeauthoryear{Engel and Fries}{}]{Engel2010}
\begin{barticle}[author]
\bauthor{\bsnm{Engel},~\bfnm{A.~K.}\binits{A.~K.}} \AND
  \bauthor{\bsnm{Fries},~\bfnm{P.}\binits{P.}}
\btitle{Beta-band oscillations -- signalling the status quo?}
\bjournal{Current Opinion in Neurobiology}
\bvolume{20}
\bpages{156-165}.
\end{barticle}
\endbibitem

\bibitem[\protect\citeauthoryear{Eriksen and Eriksen}{1974}]{Eriksen1974}
\begin{barticle}[author]
\bauthor{\bsnm{Eriksen},~\bfnm{B.~A.}\binits{B.~A.}} \AND
  \bauthor{\bsnm{Eriksen},~\bfnm{C.~W.}\binits{C.~W.}}
(\byear{1974}).
\btitle{Effects of noise letters upon the identification of a target letter in
  a nonsearch task}.
\bjournal{Percept Psychophys}
\bvolume{16}
\bpages{143-149}.
\end{barticle}
\endbibitem

\bibitem[\protect\citeauthoryear{Espana-Boquera and
  Puerta-Notario}{1996}]{EspanaBoquera1996}
\begin{barticle}[author]
\bauthor{\bsnm{Espana-Boquera},~\bfnm{M.~C.}\binits{M.~C.}} \AND
  \bauthor{\bsnm{Puerta-Notario},~\bfnm{A.}\binits{A.}}
(\byear{1996}).
\btitle{Noise effects in injection locked laser simulation: Phase jumps and
  associated spectral components}.
\bjournal{Electronics Letters}
\bvolume{32}
\bpages{818-819}.
\end{barticle}
\endbibitem

\bibitem[\protect\citeauthoryear{Fraser and Swinney}{1986}]{Fraser1986}
\begin{barticle}[author]
\bauthor{\bsnm{Fraser},~\bfnm{A.~M.}\binits{A.~M.}} \AND
  \bauthor{\bsnm{Swinney},~\bfnm{H.~L.}\binits{H.~L.}}
(\byear{1986}).
\btitle{Independent coordinates for strange attractors from mutual
  information}.
\bjournal{Phys Rev A}
\bvolume{33}
\bpages{1134-1140}.
\end{barticle}
\endbibitem

\bibitem[\protect\citeauthoryear{Gabor}{1946}]{Gabor1946}
\begin{barticle}[author]
\bauthor{\bsnm{Gabor},~\bfnm{D.}\binits{D.}}
(\byear{1946}).
\btitle{Theory of Communication}.
\bjournal{J Inst Electr Eng}
\bvolume{93}
\bpages{429-457}.
\end{barticle}
\endbibitem

\bibitem[\protect\citeauthoryear{Galambos}{1972}]{Galambos1972}
\begin{barticle}[author]
\bauthor{\bsnm{Galambos},~\bfnm{J.}\binits{J.}}
(\byear{1972}).
\btitle{On the distribution of the maximum of random variables}.
\bjournal{The Annals of Mathematical Statistics}
\bvolume{43}
\bpages{516-521}.
\end{barticle}
\endbibitem

\bibitem[\protect\citeauthoryear{Gibert and Mou\:{e}l}{2008}]{Gibert2008}
\begin{barticle}[author]
\bauthor{\bsnm{Gibert},~\bfnm{D.}\binits{D.}} \AND
  \bauthor{\bsnm{Mou\:{e}l},~\bfnm{J-L.~Le}\binits{J.-L.~L.}}
(\byear{2008}).
\btitle{Inversion of polar motion data: Chandler wobble, phase jumps, and
  geomagnetic jerks}.
\bjournal{Journal of Geophysical Research}
\bvolume{113}
\bpages{B10405}.
\end{barticle}
\endbibitem

\bibitem[\protect\citeauthoryear{Goodman}{1960}]{Goodman1960}
\begin{barticle}[author]
\bauthor{\bsnm{Goodman},~\bfnm{N.~R.}\binits{N.~R.}}
(\byear{1960}).
\btitle{Measuring Amplitude and Phase}.
\bjournal{J Franklin Inst}
\bvolume{270}
\bpages{437-450}.
\end{barticle}
\endbibitem

\bibitem[\protect\citeauthoryear{Gratton, Coles and
  Donchin}{1983}]{Gratton1983}
\begin{barticle}[author]
\bauthor{\bsnm{Gratton},~\bfnm{G.}\binits{G.}},
  \bauthor{\bsnm{Coles},~\bfnm{M.~G.~H.}\binits{M.~G.~H.}} \AND
  \bauthor{\bsnm{Donchin},~\bfnm{E.}\binits{E.}}
(\byear{1983}).
\btitle{A new method for off-line removal of ocular artifact}.
\bjournal{Electroencephalogr Clin Neurophysiol}
\bvolume{55}
\bpages{468-484}.
\end{barticle}
\endbibitem

\bibitem[\protect\citeauthoryear{Grigorenko, Nikitin and
  Kabashin}{1999}]{Grigorenko1999}
\begin{barticle}[author]
\bauthor{\bsnm{Grigorenko},~\bfnm{A.~N.}\binits{A.~N.}},
  \bauthor{\bsnm{Nikitin},~\bfnm{P.~I.}\binits{P.~I.}} \AND
  \bauthor{\bsnm{Kabashin},~\bfnm{A.~V.}\binits{A.~V.}}
(\byear{1999}).
\btitle{Phase jumps and interferometric surface plasmon resonance imaging}.
\bjournal{Applied Physics Letters}
\bvolume{75}
\bpages{3917-3919}.
\end{barticle}
\endbibitem

\bibitem[\protect\citeauthoryear{Izhikevich}{2006}]{Izhikevich_Ch10}
\begin{bincollection}[author]
\bauthor{\bsnm{Izhikevich},~\bfnm{E.~M.}\binits{E.~M.}}
(\byear{2006}).
\btitle{Synchronization}.
In \bbooktitle{Dynamical Systems in Neuroscience: The geometry of excitability
  and bursting}
\bchapter{10}
\bpages{443-505}.
\bpublisher{MIT Press}.
\end{bincollection}
\endbibitem

\bibitem[\protect\citeauthoryear{Jasper}{1958}]{Jasper1958}
\begin{barticle}[author]
\bauthor{\bsnm{Jasper},~\bfnm{H.~H.}\binits{H.~H.}}
(\byear{1958}).
\btitle{Report on the committee on methods of clinical examination in
  electroencephalography}.
\bjournal{Electroencephalogr Clin Neurophysiol}
\bvolume{10}
\bpages{370-375}.
\end{barticle}
\endbibitem

\bibitem[\protect\citeauthoryear{Kirch}{2007}]{Kirch2007}
\begin{barticle}[author]
\bauthor{\bsnm{Kirch},~\bfnm{C.}\binits{C.}}
(\byear{2007}).
\btitle{Block permutation principles for the change analysis of dependent
  data}.
\bjournal{Journal of Statistical Planning and Inference}
\bvolume{137}
\bpages{2453-2474}.
\end{barticle}
\endbibitem

\bibitem[\protect\citeauthoryear{K\"{u}nsch}{1989}]{Kunsch1989}
\begin{barticle}[author]
\bauthor{\bsnm{K\"{u}nsch},~\bfnm{H.~R.}\binits{H.~R.}}
(\byear{1989}).
\btitle{The Jackknife and the Bootstrap for General Stationary Observations}.
\bjournal{Ann. Statist.}
\bvolume{17}
\bpages{1217-1241}.
\end{barticle}
\endbibitem

\bibitem[\protect\citeauthoryear{Mart\'{\i}nez and
  Garc\'{\i}a}{2006}]{Martinez2006}
\begin{barticle}[author]
\bauthor{\bsnm{Mart\'{\i}nez},~\bfnm{N.~F.}\binits{N.~F.}} \AND
  \bauthor{\bsnm{Garc\'{\i}a},~\bfnm{R.}\binits{R.}}
(\byear{2006}).
\btitle{Measuring phase shifts and energy dissipation with amplitude modulation
  atomic force microscopy}.
\bjournal{Nanotechnology}
\bvolume{17}
\bpages{S167-S172}.
\end{barticle}
\endbibitem

\bibitem[\protect\citeauthoryear{Mead et~al.}{1992}]{Mead1992}
\begin{barticle}[author]
\bauthor{\bsnm{Mead},~\bfnm{S.}\binits{S.}},
  \bauthor{\bsnm{Ebling},~\bfnm{Francis J.~P.}\binits{F.~J.~P.}},
  \bauthor{\bsnm{Maywood},~\bfnm{E.~S.}\binits{E.~S.}},
  \bauthor{\bsnm{Humbly},~\bfnm{T.}\binits{T.}},
  \bauthor{\bsnm{Herbert},~\bfnm{J.}\binits{J.}} \AND
  \bauthor{\bsnm{Hastings},~\bfnm{M.~H.}\binits{M.~H.}}
(\byear{1992}).
\btitle{A Nonphotic Stimulus Causes Instantaneous Phase Advances of the
  Light-entrainable Circadian Oscillator of the Syrian Hamster but Does Not
  Induce the Expression of c-fos in the Suprachiasmatic Nuclei}.
\bjournal{J. Neurosci.}
\bvolume{12}
\bpages{2516-2522}.
\end{barticle}
\endbibitem

\bibitem[\protect\citeauthoryear{Osipov et~al.}{1997}]{Osipov1997}
\begin{barticle}[author]
\bauthor{\bsnm{Osipov},~\bfnm{G.~V.}\binits{G.~V.}},
  \bauthor{\bsnm{Pikovsky},~\bfnm{A.~S.}\binits{A.~S.}},
  \bauthor{\bsnm{Rosenblum},~\bfnm{M.~G.}\binits{M.~G.}} \AND
  \bauthor{\bsnm{Kurths},~\bfnm{J.}\binits{J.}}
(\byear{1997}).
\btitle{Phase synchronization effects in a lattice of nonidentical R\"{o}ssler
  oscillators}.
\bjournal{Physical Review E}
\bvolume{55}
\bpages{2353-2361}.
\end{barticle}
\endbibitem

\bibitem[\protect\citeauthoryear{Pfurtscheller}{1981}]{Pfurtscheller1981}
\begin{barticle}[author]
\bauthor{\bsnm{Pfurtscheller},~\bfnm{G.}\binits{G.}}
(\byear{1981}).
\btitle{Central beta rhythm during sensorimotor activities in man}.
\bjournal{Electroencephalography and Clinical Neurophysiology}
\bvolume{51}
\bpages{253-264}.
\end{barticle}
\endbibitem

\bibitem[\protect\citeauthoryear{Pikovsky}{1985}]{Pikovsky1985}
\begin{barticle}[author]
\bauthor{\bsnm{Pikovsky},~\bfnm{A.~S.}\binits{A.~S.}}
(\byear{1985}).
\btitle{Phase Synchronization of chaotic oscillators by a periodic external
  field}.
\bjournal{Sov. J. Commun. Technol. Electron.}
\bvolume{30}.
\end{barticle}
\endbibitem

\bibitem[\protect\citeauthoryear{Politis and Romano}{1992}]{Politis1992}
\begin{bincollection}[author]
\bauthor{\bsnm{Politis},~\bfnm{D.}\binits{D.}} \AND
  \bauthor{\bsnm{Romano},~\bfnm{J.~P.}\binits{J.~P.}}
(\byear{1992}).
\btitle{A circular block resampling procedure for stationary data}.
In \bbooktitle{Exploring the Limits of Bootstrap}
(\beditor{\bfnm{R.}\binits{R.}~\bsnm{Lepage}} \AND
  \beditor{\bfnm{L.}\binits{L.}~\bsnm{Billard}}, eds.)
\bpages{263-270}.
\bpublisher{Wiley, New York}.
\end{bincollection}
\endbibitem

\bibitem[\protect\citeauthoryear{Quyen et~al.}{2001}]{LeVanQuyen2001}
\begin{barticle}[author]
\bauthor{\bsnm{Quyen},~\bfnm{M.~Le~Van}\binits{M.~L.~V.}},
  \bauthor{\bsnm{Foucher},~\bfnm{J.}\binits{J.}},
  \bauthor{\bsnm{Lachaux},~\bfnm{J-P.}\binits{J.-P.}},
  \bauthor{\bsnm{Rodriguez},~\bfnm{E.}\binits{E.}},
  \bauthor{\bsnm{Lutz},~\bfnm{A.}\binits{A.}},
  \bauthor{\bsnm{Martinerie},~\bfnm{J.}\binits{J.}} \AND
  \bauthor{\bsnm{Varela},~\bfnm{F.~J.}\binits{F.~J.}}
(\byear{2001}).
\btitle{Comparison of Hilbert transform and wavelet methods for the analysis of
  neural synchrony}.
\bjournal{J Neurosci Methods}
\bvolume{111}
\bpages{83-98}.
\end{barticle}
\endbibitem

\bibitem[\protect\citeauthoryear{Rosenblum, Pikovsky and
  Kurths}{1996}]{Rosenblum1996}
\begin{barticle}[author]
\bauthor{\bsnm{Rosenblum},~\bfnm{M.~G.}\binits{M.~G.}},
  \bauthor{\bsnm{Pikovsky},~\bfnm{A.~S.}\binits{A.~S.}} \AND
  \bauthor{\bsnm{Kurths},~\bfnm{J.}\binits{J.}}
(\byear{1996}).
\btitle{Phase Synchronization of Chaotic Oscillators}.
\bjournal{Phys Rev Letters}
\bvolume{76}
\bpages{1804-1807}.
\end{barticle}
\endbibitem

\bibitem[\protect\citeauthoryear{R\"{o}ssler}{1976}]{Rossler1976}
\begin{barticle}[author]
\bauthor{\bsnm{R\"{o}ssler},~\bfnm{O.~E.}\binits{O.~E.}}
(\byear{1976}).
\btitle{An Equation for Continuous Chaos}.
\bjournal{Physics Letters}
\bvolume{57A}
\bpages{397-398}.
\end{barticle}
\endbibitem

\bibitem[\protect\citeauthoryear{Sporns et~al.}{2004}]{Sporns2004}
\begin{barticle}[author]
\bauthor{\bsnm{Sporns},~\bfnm{O.}\binits{O.}},
  \bauthor{\bsnm{Chialvo},~\bfnm{D.~R.}\binits{D.~R.}},
  \bauthor{\bsnm{Kaiser},~\bfnm{M.}\binits{M.}} \AND
  \bauthor{\bsnm{Hilgetag},~\bfnm{C.~C.}\binits{C.~C.}}
(\byear{2004}).
\btitle{Organization, development and function of complex brain networks}.
\bjournal{TRENDS in Cognitive Science}
\bvolume{8}
\bpages{418-425}.
\end{barticle}
\endbibitem

\bibitem[\protect\citeauthoryear{Taner, Koehler and Sheriff}{1979}]{Taner1979}
\begin{barticle}[author]
\bauthor{\bsnm{Taner},~\bfnm{M.~T.}\binits{M.~T.}},
  \bauthor{\bsnm{Koehler},~\bfnm{F.}\binits{F.}} \AND
  \bauthor{\bsnm{Sheriff},~\bfnm{R.~E.}\binits{R.~E.}}
(\byear{1979}).
\btitle{Complex seismic trace analysis}.
\bjournal{Geophysics}
\bvolume{44}
\bpages{1041-1063}.
\end{barticle}
\endbibitem

\bibitem[\protect\citeauthoryear{Wolpaw et~al.}{2000}]{Wolpaw2000}
\begin{barticle}[author]
\bauthor{\bsnm{Wolpaw},~\bfnm{J.~R.}\binits{J.~R.}},
  \bauthor{\bsnm{Birbaumer},~\bfnm{N.}\binits{N.}},
  \bauthor{\bsnm{Heetderks},~\bfnm{W.~J.}\binits{W.~J.}},
  \bauthor{\bsnm{Mc{F}arland},~\bfnm{D.~J.}\binits{D.~J.}},
  \bauthor{\bsnm{Peckham},~\bfnm{P.~H.}\binits{P.~H.}},
  \bauthor{\bsnm{Schalk},~\bfnm{G.}\binits{G.}},
  \bauthor{\bsnm{Donchin},~\bfnm{E.}\binits{E.}},
  \bauthor{\bsnm{Quatrano},~\bfnm{L.~A.}\binits{L.~A.}},
  \bauthor{\bsnm{Robinson},~\bfnm{C.~J.}\binits{C.~J.}} \AND
  \bauthor{\bsnm{Vaughn},~\bfnm{T.~M.}\binits{T.~M.}}
(\byear{2000}).
\btitle{Brain-Computer interface technology: A review of the first
  international meeting}.
\bjournal{IEEE Trans. Rehab. Eng.}
\bvolume{8}
\bpages{164-173}.
\end{barticle}
\endbibitem

\end{thebibliography}
\bibliographystyle{imsart-nameyear}

\end{document}